\documentclass[manuscript,nonacm]{acmart}
\usepackage[utf8]{inputenc}
\usepackage{color}
\usepackage{xcolor}
\usepackage{amsfonts}
\usepackage{amsmath}
\usepackage{graphicx}
\usepackage[T1]{fontenc}
\usepackage{comment}

\def\ShowComment{True}
\ifdefined\ShowComment
\def\billy#1{\marginpar{$\leftarrow$\fbox{B}}\footnote{$\Rightarrow$~{\sf #1 \textcolor{blue}{--Billy}}}}
\def\amanda#1{\marginpar{$\leftarrow$\fbox{A}}\footnote{$\Rightarrow$~{\sf #1 \textcolor{red}{--Amanda}}}}
\else
\def\billy#1{}
\def\amanda#1{}
\fi

\begin{document}
\title{Dispersion, Capacitated Nodes, and the Power of a Trusted Shepherd}
\thanks{Part of the work was done while William K. Moses Jr. was a post doctoral fellow at the University of Houston in Houston, USA. The work of W.\,K. Moses Jr. was supported in part by NSF grants CCF-1540512, IIS-1633720, and CCF-1717075 and in part by BSF grant 2016419 and in part by UMass Lowell Pre-tenure Mathematics Faculty Seed Grant. The work of A. Redlich was supported in part by UMass Lowell Pre-tenure Mathematics Faculty Seed Grant.}

\author{William K. Moses Jr.}
\email{wkmjr3@gmail.com}
\affiliation{%
  \institution{Durham University}
  \city{Durham} % ACM template requires city and
  \country{UK} % country for affiliation
}
\orcid{0000-0002-4533-7593}

\author{Amanda Redlich}
\email{amanda\_redlich@uml.edu}
\affiliation{%
  \institution{University of Massachusetts Lowell}
  \city{Lowell} % ACM template requires city and
  \country{USA} % country for affiliation
}
\orcid{0000-0002-5217-5801}

\begin{abstract}
In this paper, we look at and expand the problems of dispersion and Byzantine dispersion of mobile robots on a graph, introduced by Augustine and Moses~Jr.~[ICDCN~2018] and by Molla, Mondal, and Moses~Jr.~[ALGOSENSORS~2020], respectively, to graphs where nodes have variable capacities. We use the idea of a single shepherd, a more powerful robot that will never act in a Byzantine manner, to achieve fast Byzantine dispersion, even when other robots may be strong Byzantine in nature. We also show the benefit of a shepherd for dispersion on capacitated graphs when no Byzantine robots are present.
\end{abstract}

\begin{CCSXML}
<ccs2012>
<concept>
<concept_id>10003752.10003809.10010172</concept_id>
<concept_desc>Theory of computation~Distributed algorithms</concept_desc>
<concept_significance>500</concept_significance>
</concept>
<concept>
<concept_id>10010147.10010178.10010219.10010222</concept_id>
<concept_desc>Computing methodologies~Mobile agents</concept_desc>
<concept_significance>500</concept_significance>
</concept>
</ccs2012>
\end{CCSXML}

\ccsdesc[500]{Theory of computation~Distributed algorithms}
\ccsdesc[500]{Computing methodologies~Mobile agents}

\keywords{Mobile robots, Mobile agents, Capacitated graphs, Dispersion, Byzantine dispersion, Fault tolerance, Trusted shepherd}
\maketitle              % typeset the header of the contribution

\setcounter{footnote}{0} 
%=============================
\section{Introduction}
\label{sec:intro}

There are many instances in modern life where computational entities must move around in some space and work together with one another to perform some task. For example,  self-driving cars interacting with one another in order to navigate intersections, overtake one another, and perform other driving behavior. Another example is that of using unmanned aerial vehicles to collect information about the weather~\cite{G08}. These real world instances can be abstracted by using the model of mobile robots on a plane or on a graph. 

In the context of mobile robots on a graph, several problems have been well studied in the past, such as exploration~\cite{CFIKP08}, 
gathering~\cite{CP02}, 
scattering~\cite{EB11},  
and dispersion~\cite{AM18}.  
Of interest to us is the problem of dispersion of mobile robots on a graph, where $k$ robots, initially arbitrarily placed on an $n$ node graph, must move around such that each node has at most $\lceil k/n \rceil$ robots on it. 
This models, for example,  $k$ electric vehicles that each need a charge in a city with $n$ charging stations. Since it may take far longer for a vehicle to charge up at a station than find a station, the goal is for as few vehicles as possible to share a station.

The problem of dispersion of mobile robots on a graph, introduced in~\cite{AM18}, assumed that each node was identical with respect to satisfying the demand of the robot. However, in real life, this may not necessarily be the case. Consider again the electric vehicle example.  Some locations may have multiple chargers while others have just one.  Here we study the problem of dispersion on \emph{capacitated} graphs to model this idea.  Each node has a (possibly zero) capacity, and the total capacity across all nodes is at least the number of robots.  Now dispersion is redefined, the robots must move around such that each node has at most \emph{its capacity} of robots on it.

We further modify the problem to match real-world situations by introducing capacitated \emph{Byzantine} dispersion.  Robots can sometimes fail, either in simple ways (crash faults) or in arbitrary and unexpected ways (Byzantine faults). Byzantine dispersion was introduced and first studied in~\cite{MMM20-byzantine,MMM21-TCS}. In previous work~\cite{MMM20-byzantine,MMM21-TCS,MMM21}, it was assumed that all robots were equally likely to be corrupted. However, in real life, the robots participating in a task may be heterogeneous~\cite{RAT19}. Due to budget constraints or availability one may use a large number of inexpensive robots to perform a task and a single more expensive robot to aid these robots. 

Thus, in this paper, we look at how a single powerful robot, which we call a \textit{trusted shepherd}, and multiple weaker robots can solve Byzantine dispersion on a capacitated graph. We assume that this trusted shepherd will never act in a Byzantine manner and all robots are able to identify the shepherd. These are reasonable assumptions: The shepherd represents a special robot that may be more powerful than other robots. We can suppose that all robots have a sensor that is tuned to pick up a specific type of signal only emitted by the shepherd. For instance, if all robots are equipped with a light sensor but only the shepherd is equipped with a source of light, then it would be impossible for other robots to fake being the shepherd.\footnote{Alternatively, one may assume that the shepherd can use cryptographic primitives to ensure that he is trusted, such as encrypting messages via a private key whose corresponding public key is known to all robots.} As the shepherd represents a more powerful robot, with extra hardware or other mechanisms to prevent faulty behavior, it is reasonable to assume that the shepherd is not Byzantine.

%=======================================
\subsection{Model}
\label{subsec:model}

\textbf{Graph Description.}
We consider a graph $G(V,E,c)$ with $V$ the set of nodes, $E$ the set of edges, and $c$ a capacity function that maps each node $v \in V$ to some value in $[0,k^p]$, where $k$ is the number of robots present in the graph and $p$ is some positive constant.\footnote{Notice that some nodes may have zero capacity.} We assume that $\sum_{v \in V} c(v) \geq k$. The nodes are anonymous, i.e., they do not have IDs. We use $n$ to denote the number of nodes and $m$ to denote the number of edges. Note that neither the graph structure nor $c$ are known to robots prior to the start of the algorithm.  A robot discovers $c(v)$ only when reaching $v$.\\

\noindent \textbf{Robot Description.}
There are $k$ ($>1$) robots on the graph, each having a unique ID in the range $[1,k^q]$, where $q$ is some positive constant. Among these robots, one of them is designated the \textit{shepherd}. This robot can never be Byzantine. All robots can detect if a  co-located robot is the shepherd. 

We assume time proceeds in synchronous rounds. Each round consists of two steps: (i) robots co-located with each other communicate and each robot performs local computation, (ii) each robot stays at its current node or moves along an edge to an adjacent node. All robots start the prescribed algorithm at the same time.

Among the $k$ robots, at most $f$ of them may be strong Byzantine~\cite{DPP14}. That is, they may deviate from prescribed algorithms, send incorrect information to other robots when communicating, and lie about their IDs. (This contrasts with weak Byzantine behavior~\cite{DPP14}, where robots may \textit{not} lie about their IDs.) There is one subtlety with respect to strong Byzantine behavior that we wish to make explicit here. When the notion of a strong/weak Byzantine robot was introduced in~\cite{DPP14} and was subsequently used in~\cite{MMM21}, it was assumed that when a robot is present at a node, it could see the labels of co-located robots, and all information exchanged is done in a ``shouting'' manner so that the information becomes common knowledge to all robots co-located on the node. This implicitly prevents a Sybil style attack, i.e., a strong Byzantine robot cannot pretend to send messages originating from multiple different robots. As this is implicitly taken care of by the model, we do not explicitly handle it in our algorithms.

When we consider the Byzantine setting in Section~\ref{sec:byz-robots}, as is usual (e.g. \cite{DPP14}) we assume that all robots have unlimited memory. When we consider the non-Byzantine setting in Section~\ref{sec:vanilla-shepherd}, we assume that all robots have limited memory with the exact values required for the given algorithm. Note that it is possible for robots to not know the value of a parameter (e.g., $k$) and yet require memory that is proportional to some function of that parameter (e.g., $\log k$) in order to perform some algorithm. This is a conditional guarantee of the algorithm and so long as the robots possess this minimum memory requirement, the algorithm will work as intended.\\

\noindent \textbf{Global Knowledge and Assumptions.}
 As a thumb rule, we assume that only the shepherd knows both the values of $n$ and $k$ unless otherwise stated. Regarding the algorithms for Byzantine dispersion, we assume that the robots do not know the value of $f$ unless otherwise stated. While the explicit knowledge requirements are made clear in each Theorem, we also mention these requirements together in one place in Section~\ref{subsec:our-contrib}.\\

\noindent \textbf{Problem Statement: Dispersion on a Capacitated Graph:}
Given $k$ robots initially placed arbitrarily on a capacitated graph of $n$ nodes, the robots must re-position themselves autonomously to reach a configuration where each node $u$ with capacity $c(u)$ has at most $c(u)$ robots on it. Subsequently the robots must terminate the algorithm.\\

\noindent \textbf{Problem Statement: Byzantine Dispersion on a Capacitated Graph:}
Given $k$ robots, up to $f$ of which are Byzantine, initially placed arbitrarily on a capacitated graph of $n$ nodes, the non-Byzantine robots must re-position themselves autonomously to reach a configuration where each node $u$ with capacity $c(u)$ has at most $c(u)$ non-Byzantine robots on it. Subsequently the robots must terminate the algorithm.

%=======================================
\subsection{Useful Procedures from Other Papers}
\label{subsec:proc-other-papers}

We utilize two procedures from earlier literature. The first is the routine Explorer-Pebble, called EMT in~\cite{DPP14}, which can be used by at least two robots to construct a map of the graph, even when $n$ is unknown. Initially, two or more robots are co-located on the same node, call it the root node. One set of them acts as an explorer, and the remaining robots act as a pebble. The explorer and pebble work together to learn the map of the graph in phases, where in phase $i$, the robots learn about the subgraph consisting of all nodes within a distance of $i$ hops from the root node and all edges between these nodes. Roughly, in each phase, the explorer takes the pebble through an edge leading outside the currently known subgraph and ``places'' the pebble on that node. Subsequently, the explorer explores the rest of the known subgraph to see if that edge leads to another node within the known subgraph; the explorer also traverses all the other edges leading outside the currently known subgraph to see if any of those edges lead to the node with the pebble on it. Once this is done, the explorer ``retrieves'' the pebble and repeats this process for all edges leading out of the subgraph found from the previous phase until the connections between all nodes at distance $i$ from the root node and the nodes already discovered are found. The entire process takes $O(n^3)$ rounds.

The second procedure we utilize, the universal exploration sequence (UXS) of~\cite{R08}, is one that allows a robot to visit all nodes of a graph of size at most $n$, when $n$ is given as an input parameter. This procedure takes $n$ as an input and generates a sequence of integers, say $x_1, x_2, \ldots$. Initially, the robot leaves the current node through port $0$. Let the next series of moves be called step 1, step 2, $\ldots$. In step $i$, if the robot entered the current node $u$ through port $p_i$, then it leaves $u$ through port $p_i + x_i \mod \delta_u$, where $\delta_u$ is the degree of $u$. For any arbitrary graph, there exists a universal exploration sequence such that this procedure takes $O(n^5 \log n)$ time. To allow for various run times depending on prior knowledge of the graph (e.g., $O(d^2 n^3 \log n)$ for a $d$-regular graph~\cite{TZ14}) or for how much memory a robot has, we simply say that the procedure takes $X(n)$ rounds.

%=======================================
\subsection{Our Contributions}
\label{subsec:our-contrib}
We make several contributions in this paper. First, we formalized the  the problems of dispersion and Byzantine dispersion in the capacitated setting in this section.

In Section~\ref{sec:byz-robots}, we show how the use of a trusted shepherd to solve Byzantine dispersion leads to better guarantees. We first develop an algorithm that solves Byzantine dispersion on capacitated graphs in $O(X(n) + n^3)$ rounds and tolerates up to $\lfloor (k-1)/2 - 1 \rfloor$ strong Byzantine robots. This algorithm requires the shepherd to know the values of $n$ and $k$. We then show how to replace the assumption that the shepherd knows the value of $n$ with the assumption that the shepherd knows the value of $f$ and develop an algorithm to solve Byzantine dispersion on capacitated graphs in $O(X(n) + n^3)$ rounds that tolerates up to $(k-1)/3 - 1$ strong Byzantine robots. Finally, we develop an algorithm that replaces the use of a trusted shepherd with the assumption that the total capacity of the graph is $\geq cf + k-f$, where $c$ is the number of non-zero-capacity nodes. This algorithm solves Byzantine dispersion in $X(n)$ rounds, tolerates up to $k-1$ strong Byzantine robots, and only requires the robots to know the value of $n$, unlike the previous algorithms. However, \emph{all} robots must know the value of $n$.  

Note that an algorithm for dispersion on uncapacitated graphs which includes map creation (e.g., \cite{MMM21,MMM21-arxiv}) can be extended to capacitated graphs; simply add the capacities of nodes to the map as they are discovered. Thus, our work can immediately be compared with prior algorithms, both in terms of time and Byzantine tolerance.  We include details in a table below.

Finally, in Section~\ref{sec:vanilla-shepherd}, we show that a shepherd aids in dispersion on a capacitated graph even without Byzantine robots. In this setting we use the shepherd's memory to store a map of the whole graph.  
We explain the technical challenges that render capacitated non-Byzantine dispersion a non-trivial task.  We then develop a wrapper algorithm that (given any algorithm $\mathcal{A}$ for dispersion on uncapacitated graphs in $T_{\mathcal{A}}$ rounds and with $M_{\mathcal{A}}$ bits of memory per robot) solves dispersion on capacitated graphs in time $O(T_{\mathcal{A}} + m)$ with a requirement of $O(M_{\mathcal{A}})$ bits of memory for the non-shepherd robots and $O(M_{\mathcal{A}} + m \log k)$ bits of memory for the shepherd robot. We note that this wrapper algorithm works when $k>n$ and the values of $T_{\mathcal{A}}$, $n$, $k$, and any global knowledge parameters needed to run $\mathcal{A}$ are known to all robots.

We also develop algorithms that handle all values of $k$ with less global knowledge. If only the shepherd knows $n$ and no robot knows $k$, we give an algorithm using UXS that takes $O((X(n) + n^3 + m)$ time and $O(M_x + m \log (nk))$ bits of memory (where $M_x$ is the memory to construct and use a UXS) for the shepherd and $O(\log k)$ bits of memory for the other robots. If no robot knows $n$ or $k$, we show a $O(n^3 + m)$ time algorithm that assumes the shepherd is initially co-located with at least one other robot, the shepherd has $O(m \log (nk))$ bits of memory, and the remaining robots have $O(\log k)$ bits of memory.

We note that all our algorithms are deterministic in nature.

%========================================
\subsection{Comparison with Related Work}

The problem of dispersion of mobile robots on graphs was originally introduced in~\cite{AM18}. Subsequent work~\cite{KA19,KMS19,SSKM20,KS22} in the synchronous system focused on reducing the time-memory trade-offs to solve the problem. The current best known algorithm is that of~\cite{KS22}.

The problem has been extended to the asynchronous system~\cite{KA19,KS22}, dynamic graphs~\cite{AAMSS18,KMS20-ICDCS}, and randomized algorithms~\cite{RSSS19,MM19,DBS21} among others. 
Here we consider robots with Byzantine faults~\cite{MMM20-byzantine,MMM21-TCS,MMM21}.

In order to compare our work with prior related work, we first set our results in the same context. That is, we make the same assumption as prior work~\cite{MMM21,MMM21-arxiv} that $k=n$ (observe that in this case knowledge of $n$ and knowledge of $k$ are one and the same). Note that~\cite{MMM21-arxiv} contains the strongest results, and those results are what we compare ours with. In particular, we wish to highlight 3 of our results in this context. The algorithm in~\cite{MMM21-arxiv} handles up to $n-1$ weak Byzantine robots, where robots start from an arbitrary configuration, with the assumption that the quotient graph of the input graph is isomorphic to the input graph. We give an algorithm (Theorem~\ref{the:byz-input-condition}) that handles up to $n-1$ \textit{strong} Byzantine robots where robots start from an arbitrary configuration with the assumption that the total capacity of the input graph is $\geq cf + n-f$, where $c$ is the number of non-zero capacity nodes.   They have multiple results that do not require a restriction on the input graph and can handle robots starting from an arbitrary configuration. However, (i) their algorithms take \emph{asymptotically longer} than our algorithm (Theorem~\ref{the:byz-nk}), (ii) our algorithm has \emph{strong Byzantine fault tolerance} that is similar but slightly less than the best weak Byzantine fault tolerance of their algorithms, and (iii) when compared with their algorithm that handles strong Byzantine robots, ours \emph{requires one less parameter} to be known (i.e., $f$) and runs exponentially faster. 
A comparison of a subset of our results with prior work may be found in Table~\ref{table:results}. We achieve these improvements through use of a ``trusted shepherd''.  

To the best of our knowledge, no non-trivial time lower bounds are known for dispersion or Byzantine dispersion and a trivial lower bound of $\Omega(n)$ rounds holds.

\begin{table*}[ht]
	\caption{Comparison of our results and previous results for Byzantine dispersion of $k$ robots on an $n$ node capacitated graph in the presence of at most $f$ Byzantine robots. Note that previous results assume $k=n$ (so knowledge of $n$ implies knowledge of $k$), whereas our results are for any $k$ robots. In order to accurately compare our results with previous work, one should substitute $k=n$ in our results. The third column indicates whether all robots start at the same node (gathered) or from an arbitrary configuration. The fourth column indicates up to how many Byzantine robots the algorithm can tolerate. The fifth column indicates whether the algorithm can handle strong Byzantine robots. The final column indicates what information is required knowledge for the robots or the shepherd. Note that $\Lambda_{good}$ is the length of the largest ID among non-Byzantine robots, $\Lambda_{all}$ is the length of the largest ID among all robots, and $X(n)$ is the number of rounds required to explore any graph of $n$ nodes.}
\centering %\vspace{1em}
		\resizebox{1.0\columnwidth}{!}{%
	\begin{tabular}{|c|c|c|c|c|c|}
		\hline
		Paper &  Running Time &  Starting & Byzantine & Handles Strong & Required\\
		  &  (in rounds) &   Configuration & Tolerance &  Byzantine Robots & Knowledge  \\
		\hline
		\hline
		\cite{MMM21,MMM21-arxiv}$^\dagger$*  & $polynomial(n)$ &  Arbitrary & $n-1$ & No & $n$\\
		\hline
		\textbf{This paper (Theorem~\ref{the:byz-input-condition})$^\dagger$} &  \textbf{$X(n)$} & \textbf{Arbitrary} & \textbf{$k-1$} & \textbf{Yes} & \textbf{$n$} \\
		\hline
		\hline
		\cite{MMM21,MMM21-arxiv}$^\dagger$$^\maltese$  & $O(n^4 |\Lambda_{good}| X(n))$  & Arbitrary & $\lfloor n/2 -1 \rfloor$ & No & $n$ \\
		\hline
		\cite{MMM21,MMM21-arxiv}$^\dagger$$^\diamond$  & $O((f +  |\Lambda_{all}|) X(n))$ &  Arbitrary & $O(\sqrt{n})$ & No & $n$ \\
		\hline
		\cite{MMM21,MMM21-arxiv}$^\dagger$  & $exponential(n)$  & Arbitrary & $\lfloor n/4 - 1\rfloor$ &  Yes & $n$ and $f$  \\
		\hline
		\textbf{This paper (Theorem~\ref{the:byz-nk})$^\aleph$} & \textbf{$O(X(n) + n^3)$}  & \textbf{Arbitrary} & \textbf{$\lfloor (k-1)/2  - 1 \rfloor$} & \textbf{Yes}  & \textbf{$n$ and $k$}\\
		\hline
		\hline
		\cite{MMM21,MMM21-arxiv}$^\dagger$  & $O(n^4)$ &  Gathered & $\lfloor n/2 - 1 \rfloor$ & No & $n$\\
		\hline
		\cite{MMM21,MMM21-arxiv}$^\dagger$  & $O(n^3)$ &  Gathered & $\lfloor n/3 - 1 \rfloor$ & No & $n$\\
		\hline
		\cite{MMM21,MMM21-arxiv}$^\dagger$  & $O(n^3)$  & Gathered & $\lfloor n/4 - 1\rfloor$ & Yes & $n$\\
		\hline
		\textbf{This paper (Corollary~\ref{cor:byz-nk})}  &   \textbf{$O(n^3)$}  & \textbf{Gathered} & \textbf{$\lfloor (k-1)/2 - 1 \rfloor$} & \textbf{Yes} & \textbf{$n$ and $k$}\\
		\hline
		\hline
        \multicolumn{6}{|l|}{$^\dagger$This result assumes $k=n$.}\\
		\multicolumn{6}{|l|}{*This result holds only for those graphs where the quotient graph is isomorphic to the original graph.} \\
		\multicolumn{6}{|l|}{$^\dagger$This result holds only for those graphs where the total capacity of the graph $\geq cf + k-f$, where $c$ is the number of} \\
		\multicolumn{6}{|l|}{\hspace{1em}  non-zero capacity nodes.} \\
		\multicolumn{6}{|l|}{$^\maltese$Since $|\Lambda_{good}| = O(\log n)$ and $X(n) = \Tilde{O}(n^5)$ (see~\cite{AKLLR79,TZ14}), $O(n^4 |\Lambda_{good}| X(n)) = \Tilde{O}(n^9)$.} \\
		\multicolumn{6}{|l|}{$^\diamond$Since $|\Lambda_{all}| = O(\log n)$, $f=O(\sqrt{n})$, and $X(n) = \Tilde{O}(n^5)$ (see~\cite{AKLLR79,TZ14}), $O((f +  |\Lambda_{good}|) X(n)) = \Tilde{O}(n^5 \sqrt{n})$.} \\
		\multicolumn{6}{|l|}{$^\aleph$Since $X(n) = \Tilde{O}(n^5)$ (see~\cite{AKLLR79,TZ14}), $O(X(n) + n^3) = \Tilde{O}(n^5)$.} \\
		\hline
	\end{tabular}
		}
	\label{table:results}
\end{table*}

%=============================
\section{The Power of a Shepherd in the Land of Byzantine Robots}
\label{sec:byz-robots}

%===========================
%===========================
\subsection{Algorithm \& Analysis}
\label{subsec:byz-alg-analyis}
In this section, we present an algorithm that utilizes a trusted shepherd to allow robots to solve Byzantine dispersion on a capacitated graph, tolerating up to $\lfloor (k-1)/2 \rfloor$ strong Byzantine robots.\\

%=============================
\noindent \textbf{Brief Description.}
The algorithm is similar in structure to the algorithms from~\cite{MMM21}, i.e. it can be broken down into three stages, (i)~gathering, (ii)~map creation, and (iii)~dispersion. However, we utilize the shepherd to execute these phases differently.  In stage 1, the shepherd uses a UXS to find and gather the remaining robots. In stage 2, the gathered robots participate in an Explorer-Pebble routine, where the shepherd acts as the explorer and the remaining robots act as the pebble, to construct a map of the graph. Finally, in stage 3, the shepherd leads the other robots to find nodes to settle down at, before settling down itself. \\

%=============================
\noindent \textbf{Detailed Algorithm.}
The algorithm proceeds in 3 stages.  
In stage 1, the shepherd moves through the graph according to a UXS for $X(n)$ rounds. 
Other robots stay at their initial nodes until the shepherd passes through their node, then follow it.  When all the nodes have been explored (so all the robots are gathered), the shepherd ask for IDs of all robots and remembers them. 

In stage 2, the robots perform the Explorer-Pebble routine.  The shepherd acts as the explorer while the other robots move as a group to act as the pebble. Importantly, the shepherd only treats the set of robots it finds on a node as a pebble when the size of the set is at least $\lceil (k-1)/2 \rceil$.  

In stage 3, the shepherd takes the robots on a depth first search of the spanning tree of the graph and tell robots where to settle down. The shepherd decides where to settle robots as follows.  
The shepherd maintains a list of which robots have been settled, also a list of which nodes have been allotted robots.  At each node $u$ that the shepherd visits (including the initial node it starts this stage on) with some non-zero capacity $c(u)$, the shepherd assigns the at most $c(u)$ lowest ID robots (excluding itself) that have not yet been settled to that node. The shepherd also updates its list of IDs that have been settled. Note that strong Byzantine robots may fake their IDs initially, resulting in multiple robots claiming to have the same ID. In this case, the shepherd treats all robots with the same claimed ID the same and tells all of them to sit at the same node $u$ when that ID falls within the $c(u)$ lowest IDs. After all IDs, excluding itself, have been settled, the shepherd moves to a node with sufficient capacity and settles down as well. When a robot is told to settle at a given node, it does so and terminates the algorithm. Once the shepherd settles at a node, it also terminates the algorithm.\\

%=============================
\noindent \textbf{Analysis.}

\begin{theorem}\label{the:byz-nk}
There exists an algorithm that allows $k$ robots, up to $f$ of which are strong Byzantine robots, that are initially arbitrarily located on an $n$ node capacitated graph to solve Byzantine dispersion in $O(X(n) + n^3)$ rounds when $f < \lfloor (k-1)/2 \rfloor$. Furthermore, $k$ and $n$ must be known to the shepherd.
\end{theorem}

\begin{proof}
Let us first analyze the running time. Stage 1 completes in $X(n)$ rounds. Stage 2 takes $O(n^3)$ rounds. Finally, stage 3 takes $O(n)$ rounds, resulting in time complexity $O(X(n)+n^3)$.

As for correctness, we argue the following three things: (i) at the end of stage 1, all non-Byzantine robots are gathered together, (ii) at the end of stage 2, the shepherd has an accurate map of the graph, and (iii) at the end of stage 3, Byzantine dispersion is achieved.

To argue (i), notice that the non-Byzantine robots that are not the shepherd do not move until they are collected by the shepherd. The shepherd is guaranteed to visit every node in the graph due to running a UXS sequence parameterized by $n$. After $X(n)$ rounds, it is guaranteed that the shepherd would have collected all robots that stayed put, i.e., all non-Byzantine robots.

To see that (ii) holds, notice that $f< \lfloor (k-1)/2 \rfloor$. Since the shepherd won't treat a group of robots as a pebble if the size is $< \lceil (k-1)/2 \rceil$, it cannot be fooled by the Byzantine robots. Thus, at the end of stage 2, the shepherd will have an accurate map of the graph.

To see that (iii) holds, notice that robots are settled as per their IDs. At the beginning of stage 3, the shepherd takes stock of the claimed IDs of all robots. 
Since the robot IDs provided by non-Byzantine robots are unique with respect to one another, when the shepherd indicates which IDs should settle down at a given node, the number of non-Byzantine robots that settle down at that node will not exceed the capacity of the node. Thus, at the end of stage 3, Byzantine dispersion is achieved.
\end{proof}

Notice that if robots are initially gathered, then they only need to run stage 2 and stage 3 of the algorithm, resulting in a faster run time. This is also an improvement on the total memory requirements of robots compared with~\cite{MMM21}; here only one robot, the shepherd, creates and stores the map compared with all robots creating and storing maps in~\cite{MMM21}.

\begin{corollary}\label{cor:byz-nk}
There exists an algorithm that allows $k$ robots, up to $f$ of which are strong Byzantine robots, to solve Byzantine dispersion on an $n$ node graph in $O(n^3)$ rounds when $f < \lfloor (k-1)/2 \rfloor$ and all robots are initially gathered at the same node. Furthermore, $k$ and $n$ must be known to the shepherd.
\end{corollary}

%=============================
%===========================
\subsection{Replacing the Shepherd's Knowledge Requirement of $n$ with $f$}
\label{subsec:byz-replace-n-f}
In this section, we design an algorithm that substitutes the shepherd's knowledge requirement of $n$ with that of $f$. However, the tolerance of the algorithm to Byzantine robots is reduced, i.e., the algorithm can only handle $f< (k-1)/3$ strong Byzantine robots.\\

%=============================
\noindent \textbf{Brief Description.}
This is a three stage process. In stage one, all non-shepherd robots perform UXSes with input parameter $2^i$ for increasing values of $i=1,2,3,\ldots$ until they find the shepherd. The shepherd waits until at least $k-f-1$ other robots find it and then moves to stage two. In stage two, the shepherd and the gathered robots construct the map of the graph using an explorer-pebble routine, resulting in the shepherd knowing the value of $n$. The final stage consists of the shepherd waiting at a node for a sufficient amount of time to allow for  any remaining non-Byzantine robots find it, then settling the robots and itself as before.\footnote{Recall that $f$ is only an upper bound on the number of Byzantine robots and so there may be more non-Byzantine robots that did not find the shepherd yet.}\\

%=============================
\noindent \textbf{Detailed Algorithm.}
The algorithm proceeds in three stages corresponding to (i) gathering, (ii) calculating the value of $n$, and (iii) waiting for any remaining non-Byzantine robots to find the shepherd.

In stage one, every robot that is not a shepherd runs UXSes with input parameter $2^i$ for increasing values of $i=1,2,3, \ldots$ until the robot finds the shepherd. Once the robot finds the shepherd, it follows the shepherd. The shepherd waits until $k-f-1$ other robots find it and then moves to stage two.\footnote{Note that it may be the case that more than $k-f-1$ robots find the shepherd and follow it if several robots find the shepherd in the same round.}

In stage two, the shepherd first asks the (at least $k-f-1$) other robots for their IDs and then performs an Explorer-Pebble routine to construct a map of the graph. The shepherd acts as the explorer and the remaining robots act as the pebble. To ensure that the shepherd is not fooled by Byzantine robots pretending to be a pebble, the shepherd ignores groups of $<k-2f-1$ robots and only considers a group of at least $k-2f-1$ robots with IDs belonging to the set of IDs collected by it at the beginning of stage two to be the pebble. At the end of stage two, the shepherd has the map and knows the value of $n$. Note that if a non-shepherd robot, that did not find the shepherd in stage one, finds the shepherd in stage two, it simply follows the shepherd until the shepherd asks it to terminate the algorithm in stage three.

At the beginning of stage three, the shepherd knows the current round number and the value of $n$. From these two values, the shepherd calculates the round number $r$ by which any robot that is still performing UXS can start and finish a complete UXS with some input parameter $2^j$ such that $2^j \geq n$ and thus find the shepherd. The shepherd waits for $r$ steps, then leads the other robots to find nodes to settle down at, before settling down itself.\\

%=============================
\noindent \textbf{Analysis.}
\begin{theorem}\label{the:byz-kf}
There exists an algorithm that allows $k$ robots, up to $f$ of which are strong Byzantine robots, that are initially arbitrarily located on an $n$ node capacitated graph to solve Byzantine dispersion in $O(X(n) + n^3)$ rounds when $f <  (k-1)/3 $. Furthermore, $k$ and $f$ must be known to the shepherd.
\end{theorem}

\begin{proof}
We first argue correctness and then run time. At the end of stage one, the shepherd is present with at least $k-f-1$ other robots. In stage two, the only issue that may arise is that a set of Byzantine robots may try to fool the shepherd into thinking that they are the pebble. This is circumvented since of the at least $k-f-1$ robots that can act as the pebble, of which $f$ of them may be Byzantine robots, it is the case that at least $k-2f-1$ of them are non-Byzantine (since $f<(k-1)/3$). Thus, there are always enough non-Byzantine robots to form the pebble. Furthermore, to avoid issues where the shepherd may mistake another set of good robots with the robots acting as a pebble, we have the shepherd record the IDs of all the robots that gathered with it at the beginning of stage two. Since robot IDs are unique and there are $f< (k-1)/3)$ strong Byzantine robots, it cannot be the case that the Byzantine robots can work together with another set of good robots to pretend to be the pebble. 

Now we argue about the run time. Stage one takes at most $O(X(n))$ rounds to complete as when $2^i > n$, it is guaranteed that every non-Byzantine robot can explore the entire graph using the associated UXS. Stage two takes $O(n^3)$ rounds. Finally, for stage three, note that the total amount of time that stage one and stage two combined take is $O(X(n) + n^3)$. For any non-Byzantine robots still running a UXS, let the current input parameter to the UXS be some $2^j$. There are two cases, either $X(n) = o(n^3)$ or $X(n) = \Omega(n^3)$. In the former case, the time taken for the next UXS will be some $O(n^3)$, in which case, the time until stage three is done is $O(X(n) + n^3)$. In the latter case, the time taken for the next UXS will be some $X(O(n)) = O(X(n))$, in which case the time until stage three is done is $O(X(n) + n^3)$. The time taken for the shepherd to lead the other robots to settle down, once it has the map, is subsumed in that value. Thus, the total time until the algorithm is terminated is $O(X(n) + n^3)$ rounds.
\end{proof}

%=============================
%===========================
\subsection{Replacing the Shepherd with an Input Condition}
\label{subsec:replace-shepherd}

In this section, we describe an algorithm that allows robots to achieve Byzantine dispersion on a capacitated graph without the use of a shepherd and that tolerates up to $k-1$ strong Byzantine robots, so long as the total capacity of the graph is greater than or equal to $cf + k -f$, where $c$ is the number of non-zero capacity nodes. Furthermore, unlike the previous two algorithms, only the knowledge of $n$ is required (not $f$ and not $k$), but it must be known to all the robots.\\

%=============================
\noindent \textbf{Brief Description.}
Each robot $R$ explores the graph according to the universal exploration sequence with input parameter $n$ for $X(n)$ rounds.  At each timestep, at each node, the robots present check its remaining capacity and then that many robots settle there while the remainder (if any) keep exploring. After $X(n)$ rounds are over, the robot terminates the algorithm. \\

%=============================
\noindent \textbf{Detailed Algorithm.}
Consider some robot $R$. In a given round $i$, $1 \leq i \leq X(n)$, $R$ does the following depending on whether it has settled at the given node $u$ with capacity $c_u$ or not. If $R$ has not already settled at $u$ in a previous round, then it communicates with any other robots that may be present at the node to determine the set of robots already settled at that node, $S(u)$, and the set of robots that currently want to settle at the node in this round, $W(u)$. $R$ communicates with the other robots in $W(u)$ to determine its rank, $rank_R$, based on the ordering of the IDs of the robots in the set. If $|S(u)| + rank_R \leq c_u$, then $R$ settles at the node. Else if $|S(u)| + rank_R > c_u$, then $R$ moves to a new node according to the UXS. Note that if a strong Byzantine robot attempts to imitate more than one robot, only the highest ID is used to calculate its rank. (Recall that the model restricts strong Byzantine robots from successfully impersonating two different robots. Thus, even if the robot pretends to send messages originating from two different robots, the other robots are aware of this and can connect both messages back to the same robot.)

If $R$ settles at some node $u$ in some round $i$, then in subsequent rounds, it stays at $i$ and communicates with other robots as needed.

Finally, in round $X(n)$, all robots terminate the algorithm.\\

%=============================
\noindent \textbf{Analysis.}
\begin{theorem}\label{the:byz-input-condition}
There exists an algorithm that allows $k$ robots who all know $n$, up to $f$ of which are strong Byzantine robots, that are initially arbitrarily located on an $n$ node capacitated graph, to solve Byzantine dispersion in $X(n)$ rounds when $f < k$ and the total capacitance of the graph is $\geq cf+k - f$, where $c$ is the number of non-zero capacity nodes. %Furthermore, $n$ must be known to all robots.
\end{theorem}

\begin{proof}
The time complexity of the algorithm is easy to see as all robots must terminate the algorithm in $X(n)$ rounds. As for correctness, we argue that each non-Byzantine robot $R$ eventually finds a node $u$ such that it can settle at that node. By the condition of the initial graph, we see that each node that has a non-zero capacity has capacity at least $f$. That is, each node has enough capacity to host all $f$ Byzantine robots at each node. Additionally, there is enough additional capacity ($k-f$) distributed among the $c$ nodes to allow all the non-Byzantine robots to settle at those nodes even if the at most $f$ Byzantine robots were to be present at each node whenever a non-Byzantine robot visits it. By the properties of a UXS, it is guaranteed that $R$ will eventually visit all nodes in time $X(n)$. Thus, $R$ will eventually find a node to settle down at.
\end{proof}

%=============================
\section{The Use of a Shepherd in the Capacitated Model with no Byzantine Robots}
\label{sec:vanilla-shepherd}
In this section, we show how a trusted shepherd can be beneficial even in the absence of Byzantine robots. In Section~\ref{subsec:vanilla-alg-analysis}, we present a wrapper algorithm that allows one to solve dispersion on a capacitated graph using any pre-existing dispersion algorithms for an uncapacitated graph. We make the assumption that $k \geq n$ for this algorithm, which is natural  (e.g., consider cars navigating a city in search of parking where roads are edges and garages are nodes). In Section~\ref{subsec:vanilla-handling-klessn}, we show how to remove this assumption at the expense of a possibly larger run time and possibly more memory required by the shepherd.\\

\noindent \textbf{Motivation and Technical Challenges:} 
A typical dispersion algorithm would explore the graph by settling robots on vacant nodes as they are encountered, and using those settled robots as signposts for unsettled robots.  That process will not work for the capacitated version: Suppose we have to disperse $k$ robots on an $n$ node graph where only two nodes have non-zero capacity.   We have a contradiction once the first but not the second non-zero-capacity node is found.  The robots should settle there (and leave their zero-capacity nodes unsettled), but on the other hand the robots must remain on the zero-capacity nodes to facilitate exploring the rest of the graph and finding the second non-zero capacity node.

The situation becomes even more difficult if the number of robots $k$ is less than the number of nodes $n$, in which case the robots may not even have enough memory collectively to store a map of the graph. Our algorithm successfully overcomes both of these issues.\\

\noindent \textbf{Notes and Assumptions:} We note that pre-existing algorithms for dispersion often assume that either $k=n$ or $k<n$, but those algorithms can be easily converted into algorithms for $k \geq n$ when the value of $\lceil k/n \rceil$ is known to the robots: Instead of assuming only $1$ robot may settle at a node, they may assume that up to $\lceil k/n \rceil$ robots may settle at each node and act appropriately. An additional $O(\log (k/n))$ bits of memory per robot is needed, however that is subsumed by the $O(\log k)$ bits for each robot to store its own unique ID.

We assume that all robots have access to a dispersion algorithm $\mathcal{A}$, know its run time $T_{\mathcal{A}}$ and have enough memory $M_{\mathcal{A}}$ to run $\mathcal{A}$. We require that all robots know $n$ and $k$ and also whatever assumptions are needed in order to run $\mathcal{A}$.

\subsection{Handling $k \geq n$ Robots}
\label{subsec:vanilla-alg-analysis}
%=============================
As discussed above, the primary challenges for dispersion across a capacitated graph are creating a map of the graph and then determining where each robot should settle.  We use a shepherd to deal with both difficulties.\\

\noindent \textbf{Brief Description.}
Our wrapper algorithm consists of three phases. In phase one, all robots run an already existing dispersion algorithm $\mathcal{A}$, treating the graph as uncapacitated. In phase two, the shepherd performs a depth first search (DFS) of the graph to construct a map of the graph using the temporarily settled robots to differentiate nodes. In phase three, the shepherd first collects all the robots using a DFS  as before, then performs a second DFS allocating robots to nodes subject to capacity constraints. When the shepherd allots a robot to a node, that robot terminates the algorithm.  Once all robots (including the shepherd) are allotted to nodes, the shepherd terminates.\\

\noindent \textbf{Detailed Algorithm.}
The algorithm has three phases.  In the first phase, all robots disperse over the graph using a pre-existing dispersion algorithm $\mathcal{A}$ with run time $T_{\mathcal{A}}$ for uncapacitated graphs as a black box. Each robot implements the algorithm naively and does not pay attention to the capacity of the nodes. 
At the end of $T_{\mathcal{A}}$ rounds, all robots proceed to phase two. 

In the second phase, the non-shepherd robots remain where they were at the end of phase one, while the shepherd constructs a map of the graph (adjacency list representation of the graph) using a DFS. In particular, each non-shepherd robot communicates with the shepherd when they are co-located so that the shepherd can use the ID of the lowest ID robot residing at a node as an ID for that node. In this manner, the shepherd can perform a DFS over the graph and record the connections between nodes as well as the capacities of various nodes. Note that if $k=n$, then there will be exactly one node without a non-shepherd robot residing at it; the shepherd can still accurately construct the map of the graph using its own ID to designate that node. At the end of the DFS, the shepherd knows that the map of the graph is constructed and moves to phase three. Note that the non-shepherd robots are not immediately aware that the shepherd has moved to phase three.

In the third phase, the shepherd first performs a DFS of a spanning tree of the graph (using the map constructed in phase two) and collects all the other robots. In particular, when the shepherd reaches a node with some robots that have been there from phase one, the shepherd tells these robots that it is now in phase three and those robots subsequently follow the shepherd wherever it goes. Once the remaining $k-1$ robots have been collected and are following the shepherd, it starts a new DFS and tells the remaining robots where to settle before finally settling at some node. In particular, the shepherd keeps track of which nodes have already had robots settled at them and when the shepherd and other robots reach a new node $u$ with capacity $c(u)$, the shepherd instructs at most $c(u)$ \textit{other} robots with the smallest IDs among the set of robots traveling together to settle down at the node and terminate the algorithm. (Even if the shepherd has a smaller ID than some of the other robots, it does not settle down until all the remaining robots are settled.) Once only the shepherd remains, it settles down at the first available node with sufficient capacity and also terminates. \\

Theorem~\ref{vanilla}, given below, captures the properties of the algorithm.

\begin{theorem}\label{vanilla} Assume that $k$ robots have access to an algorithm $\mathcal{A}$ that solves dispersion on an uncapacitated graph of $n$ nodes in time $T_{\mathcal{A}}$ and requires each robot to have $M_{\mathcal{A}}$ bits of memory. If $k \geq n$ and all the robots know the value of $T_{\mathcal{A}}$, $n$, $k$, and whatever other global knowledge is required to run $\mathcal{A}$, then there exists an algorithm to solve dispersion of $k$ robots on an $n$ node capacitated graph in time $O(T_{\mathcal{A}}+m)$, where $m$ is the number of edges of the graph, with a memory requirement of $O(M_{\mathcal{A}} + m \log k)$ bits of memory for a shepherd robot and $O(M_{\mathcal{A}})$ bits of memory for the remaining $k-1$ robots.
\end{theorem}

\begin{proof}
The running time is simple to see. The first phase takes $T_{\mathcal{A}}$ time, the second phase takes $O(m)$ time, and the third phase takes $O(n)$, resulting in a total run time of $O(T_{\mathcal{A}} + m)$.

As for the memory requirement, the non-shepherd robots only need enough memory run the blackbox dispersion algorithm $\mathcal{A}$. The shepherd, in addition to this, needs sufficient memory to maintain a map of the graph. An adjacency list representation of the graph requires $O(m \log k)$ bits corresponding to $O(\log k)$ bits for each of the $m$ edges since the nodes take on IDs corresponding to the $O(\log k)$ bits IDs of the robots. In this map, for each of the $n$ nodes, we additionally need $O(\log k)$ bits to represent its capacity.

As for the correctness of the algorithm, we prove the following: (i) at the end of phase one, each node has at least one robot settled on it, (ii) at the end of phase two, the shepherd has a map of the graph, and (iii) at the end of phase three, dispersion on the capacitated graph has been achieved. In phase one, since all robots execute a pre-existing dispersion algorithm and since $k \geq n$, it is easy see that at the end of the phase, each node has at least one robot on it. In phase two, either all nodes will have a non-shepherd robot settled on them or there will exist exactly one node without a non-shepherd robot settled on it. In either case, each node can be uniquely identified and so that the shepherd may conduct a DFS to accurately map the graph. In phase three, the shepherd first collects all the robots using a DFS of a spanning tree of the graph and then tells these robots where to settle before settling itself using another DFS of a spanning tree of the graph. Thus at the end of phase three, dispersion on the capacitated graph is achieved.
\end{proof}

The current best known algorithm for dispersion on an uncapacitated graph is that of~\cite{KS22}, which allows $k$ robots initially arbitrarily located on an $n$ node graph (where $k \leq n$) to achieve dispersion in $O(\min \lbrace m, k \Delta \rbrace)$ time, where $\Delta$ is the maximum degree of the graph and $m$ is the number of edges, and requires each robot to have $O(\log (k + \Delta))$ bits of memory. From the previous theorem and our discussion just prior to Section~\ref{subsec:vanilla-alg-analysis} on how to convert dispersion algorithms on uncapacitated graphs that work for $k \leq n$ to dispersion algorithms on uncapacitated graphs that work for $k \geq n$, we have the following corollary.

\begin{corollary}
Assume that $k$ robots have access to the dispersion algorithm $\mathcal{A}$ from~\cite{KS22} for uncapacitated graphs. If $k \geq n$ and all the robots know the values of $O(\min \lbrace m, k \Delta \rbrace)$ (where $m$ is the number of edges of the graph and $\Delta$ is the maximum degree of the graph), $n$, $k$, and the global knowledge requirements of algorithm $\mathcal{A}$, then there exists an algorithm to solve dispersion of $k$ robots on an $n$ node capacitated graph in time $O(m)$ with a memory requirement of $O(m \log k)$ bits of memory for the shepherd robot and $O(\log (k + \Delta))$ bits of memory for the remaining $k-1$ robots.
\end{corollary}

%===================
\subsection{Handling Any Value of $k$ Robots}
\label{subsec:vanilla-handling-klessn}
The algorithm in Section~\ref{subsec:vanilla-alg-analysis} assumed that the number of robots $k$ was greater than or equal to the number of nodes $n$. Here, we describe a modification of the preceding algorithm to deal with scenarios where $k$ is less than $n$ or $k$ is unknown. To handle these situations, we use a UXS or the Explorer-Pebble routine, both described in Section~\ref{subsec:proc-other-papers}, rather than a pre-existing dispersion algorithm. In the below descriptions, we use $M_x$ to denote the memory required by a robot to construct and use a UXS.
\\
\\
\noindent \textbf{Brief Description.}
First, the shepherd runs a UXS with input parameter $n$ until it finds another robot. Next, these two robots perform the Explorer-Pebble routine where the shepherd acts as the explorer to construct a map of the graph which is stored with the shepherd. The shepherd learns the value of $k$ while constructing the map by noting how many robots are at each node.  After mapping, the shepherd runs phase three of the algorithm from Section~\ref{subsec:vanilla-alg-analysis} to collect and then disperse the robots. 

Note that in this algorithm, only the shepherd needs to know the value of $n$, i.e., $n$ does not need to be global knowledge, and $k$ can be unknown. Also, the non-shepherd robots only need $O(\log k)$ bits of memory each to store their unique IDs. Since $k$ may be less than $n$, the shepherd's memory required for the map is $O(m \log (nk))$.

\begin{theorem}
There exists an algorithm that allows $k$ robots that are initially arbitrarily located on an $n$ node capacitated graph with $m$ edges to solve dispersion in time $O(X(n) + n^3 + m)$, requires the shepherd to have $O(M_x + m \log (nk))$ bits of memory, where $M_x$ is the memory required to construct and use a UXS and $X(n)$ is the running time of that procedure, and the remaining robots to have $O(\log k)$ bits of memory. Only the shepherd needs to know the value of $n$.
\end{theorem}
\begin{proof} Follows from algorithm description and proofs of Theorems 1 and 4.
\end{proof}

Notice that the running time may be large due to the UXS. If we assume that all robots are initially gathered, or that the starting configuration is such that at least one other robot starts on the same node as the shepherd, then we may skip the use of the UXS and directly move to the use of the Explorer-Pebble routine to construct the map of the graph. In this scenario, none of the robots need to know either the value of $n$ or $k$.

\begin{theorem}
There exists an algorithm that allows $k$ robots that are initially located such that the shepherd and at least one robot initially occupy the same node on an $n$ node capacitated graph with $m$ edges to solve dispersion in time $O(n^3 + m)$ and requires the shepherd to have $O(m \log (nk))$ bits of memory and the remaining robots to have $O(\log k)$ bits of memory. None of the robots need to know the value of $n$ or $k$.
\end{theorem}

\begin{proof} Follows directly from Theorem 5.
\end{proof}

%=============================
\section{Conclusion}
\label{sec:conclusion}

This work suggests several new lines of future research. Here we analyzed dispersion on graphs where nodes have different capacities; what if the robots have different needs as well? This is a model for when some robots need more of a given resource (e.g., space to park, energy at a charging station, etc.).  Coupled with capacitated graphs, this would be a more realistic analogue to many real-world situations.

Another line of research is exploring the utility of a trusted shepherd in other variants of the traditional models for dispersion and Byzantine dispersion. For instance, when dealing with dynamic graphs, could a slightly more powerful robot aid the algorithm designer? What power would be most helpful?

A third  line of research is to see whether the algorithms in this paper may be adapted to an asynchronous setting. Some of the algorithms, especially those for Byzantine dispersion, may appear to be asynchronous in nature as they are event driven (as opposed to time dependent) and the shepherd is the robot that tells other robots to terminate the algorithm. However, a key component of those algorithms is the communication inherent in each round. If there were no rounds, then it becomes difficult to tell whether all robots co-located at a node had the chance to communicate with one another. It especially becomes difficult to differentiate between a relatively ``slow'' robot and one that is acting in a Byzantine manner by refusing to communicate.  Perhaps the additional power of the trusted shepherd could be used to overcome these challenges.

 %\bibliographystyle{splncs04}
\bibliographystyle{ACM-Reference-Format}
 %\bibliography{references}
 \bibliography{references-conf-short-names}

\end{document}